\documentclass[preprint,review,12pt]{elsarticle}
\usepackage{amssymb}
\usepackage{graphicx}
\usepackage{color}
\usepackage{soul}
\usepackage{xspace}

\newtheorem{theorem}{Theorem}
\newtheorem{lemma}{Lemma}
\newtheorem{observation}{Observation}
\newtheorem{definition}{Definition}
\newproof{proof}{Proof}

\newcommand{\B}{\ensuremath{\mathcal{B}}}
\newcommand{\C}{\ensuremath{\mathcal{C}}}
\newcommand{\G}{\ensuremath{\mathcal{G}}}
\newcommand{\graph}{\G}

\newcommand{\M}{\mathcal{M}}
\newcommand{\obs}{\mathcal{O}}

\newcommand{\smallpar}[1]{{\smallskip\noindent\bfseries #1}}

\def\etal{{et~al.{}}}
\newcommand{\Reals}{{\mathbb{R}}}

\def\dist{\mathbf{d}}

\def\distg{\dist_{\graph}}

\newcommand{\eps}{\varepsilon}               

\usepackage[nodots,nocompress]{numcompress}

\biboptions{super}

\begin{document}

\begin{frontmatter}

\title{Geodesic  Spanners for  Points in $\Reals^3$ amid  Axis-parallel Boxes }

\author[sut]{Mohammad Ali Abam}
\ead{abam@sharif.edu}

\author[sut]{Mohammad Javad Rezaei Seraji}
\ead{mjrezaei@ce.sharif.edu}

\address[sut]{Department of Computer Engineering, Sharif University of Technology, Tehran, Iran}

\begin{abstract}

Let $P$ be a set of $n$ points in $\Reals^3$ amid a bounded number of obstacles. When obstacles are axis-parallel boxes, we prove that $P$ admits an $8\sqrt{3}$-spanner with $O(n\log^3 n)$ edges with respect to the geodesic distance. This is the first geodesic spanner for points in $\Reals^3$ amid obstacles.
\end{abstract}


\end{frontmatter}

\section{Introduction}

When designing a network---like a road or a railway network---our main desire is to have a sparse network in which there is a fast connection between any pair of nodes in the network.  
This leads to the concept of \emph{spanners}, as defined next.

In an abstract setting, a metric space $\M=(P, \dist_\M)$ is given, where the
points of  $P$  represent the nodes in the
network, and $\dist_\M$ is a metric on~$P$. A \emph{$t$-spanner} for $\M$,
for a given $t>1$, is an edge-weighted graph $\graph=(P, E)$ where the weight
of each edge $(p,q)\in E$ is equal to $\dist_\M(p,q)$, and for all pairs $p,q\in P$ we have that $\distg(p,q)\leq t\cdot \dist_\M(p,q)$,
where $\distg(p,q)$ denotes the length  of a shortest path (that is, minimum-weight) 
from $p$ to $q$ in $\graph$ (the distance between $p$ and $q$ in $\graph$, for short).
Indeed, the distance between any two points
in the spanner $\graph$ approximates their original distance in the metric space $\M$
up to a factor~$t$. Any path from $p$ to $q$ in $\graph$ whose weight is at most $t \cdot \dist_\M(p,q)$ is called
a $t$-path.
The \emph{spanning ratio} (or dilation,
or stretch factor) of~$\graph$ is the minimum 
$t$ for which~$\graph$ is a $t$-spanner for the metric space $\M$.
The size of~$\graph$ is defined as the number of edges in $E$. The question
now becomes: given a desired spanning ratio $t$, how many edges do we need to obtain a
$t$-spanner?

\smallpar{Previous work.}
When the metric space $\M$ does not have any additional properties, one can get a $(2k-1)$-spanner of size $O(n^{1+1/k})$,
for any integer~$k>0$ by the method given in \cite{addjj-osswg-93} and an improvement on its main lemma (Lemma 6 in \cite{addjj-osswg-93}) in \cite{alon2002}. No methods are known to obtain a constant-spanning-ratio  spanner of size~$O(n\ \textrm{polylog}\ n)$ in general metric spaces. However, for several special
types of metric spaces, better results can be obtained. We next mention some of them.

When $\M$ is the \emph{Euclidean metric} (i.e., $P$ is a set of $n$  points  in $\Reals^d$ and the Euclidean distance is used), for any fixed $\eps>0$ one can then obtain a
$(1+\eps)$-spanner with $O(n)$ edges--see the book by Narasimhan and Smid~\cite{ns-gsn-07} for
fundamental results on geometric spanners. This  result \cite{cgmz-ohrdm-05, cg-sdpss-06,gr-odsdms-08, hm-fcnld-05, t-bealdm-04} was generalized to metric spaces of \emph{bounded  dimension}  (a metric space $\M=(P, \dist_\M)$ has
doubling dimension $d$ if any ball of radius $r$ in the space can be covered
by $2^d$ balls of radius~$r/2$) 

Recently, Abam \etal \cite{abam2019geodesic} showed that a set $P$ of $n$  points on a polyhedral terrain  admits a $(2+\eps)$-spanner with $O(n\log n)$ edges. This improved two recent results that deal with special cases
of geodesic spanners on terrains, namely additively weighted spanners \cite{abfgs-gswps-11} and spanners for points in a polygonal domain with some holes \cite{a-gspipd-18}.

\smallpar{The problem definition.} 
Let $P$ be a set of $n$ points in $\Reals^3$ amid a bounded number of disjoint obstacles (they do not even have a common boundary) 
and each obstacle is an axis-parallel box. Now, consider the metric space $\M=(P, \dist_\M)$ where $\dist_\M(p,q)$ is the geodesic distance of $p$ and $q$, i.e., the length of a shortest path from $p$ to $q$ avoiding obstacles. The goal is to construct for $\M$,  a $t$-spanner of size~$O(n\ \textrm{polylog}\ n)$  for some constant $t$. Note that we desire  to have a  spanner whose size is independent of the number of obstacles and indeed we will construct a spanner for $\M=(P, \dist_\M)$  while we are not allowed to use any Steiner points like the vertices of obstacles.  In other words, let $S$ be a complete graph on $n$ vertices where each node corresponds to a point $p\in P$. For nodes $u$ and $v$ of $S$  corresponding to $p,q \in P$,  the edge $(u,v)$  is associated with the geodesic distance of $p$ and $q$ as its weight. We indeed want to compute a near linear-size $t$-spanner $\graph$ for  $S$.

\smallpar{Our results.}  
When obstacles are convex $\alpha$-fat objects \cite{Katz} (not necessarily axis-parallel boxes), the geodesic distance of any two points is at most a constant factor (depending on $\alpha$) of their Euclidean distance \cite{paul2002}. Therefore, $\M$ is a metric space of constant doubling dimension and consequently, it has a $(1+\eps)$-spanner of size $O(n)$. When obstacles do not have additional properties, it is unknown how to construct a $t$-spanner of size~$O(n\ \textrm{polylog}\ n)$  for some constant $t$.
In this paper, we present an $8\sqrt{3}$-spanner of size~$O(n\log^3 n)$ when obstacles are axis-parallel boxes.

\section{A near linear-size spanner}
Suppose $P = \{p_1, \dots, p_n\}$ is a set of $n$ points in $\Reals^3$ amid some axis-parallel boxes as obstacles. For $p, q \in \Reals^3$ being outside obstacles, let $\sigma(p,q)$  be the geodesic distance of $p$ and $q$ (i.e., the length of a shortest path from $p$ to $q$ avoiding obstacles). Also, let $\B(p,q)$ be an axis-parallel box whose two opposite corners are $p$ and $q$ as depicted in  Fig. \ref{bpq} (i.e.,$p$ and $q$ are the endpoints of one of the $4$ main diagonals of  $\B(p,q)$). For ease of presentation, from now on, we assume that any point that we use, is outside or on the boundary of obstacles unless it is explicitly mentioned otherwise. 

\begin{figure*}
	\begin{center}
		\includegraphics[scale=0.4]{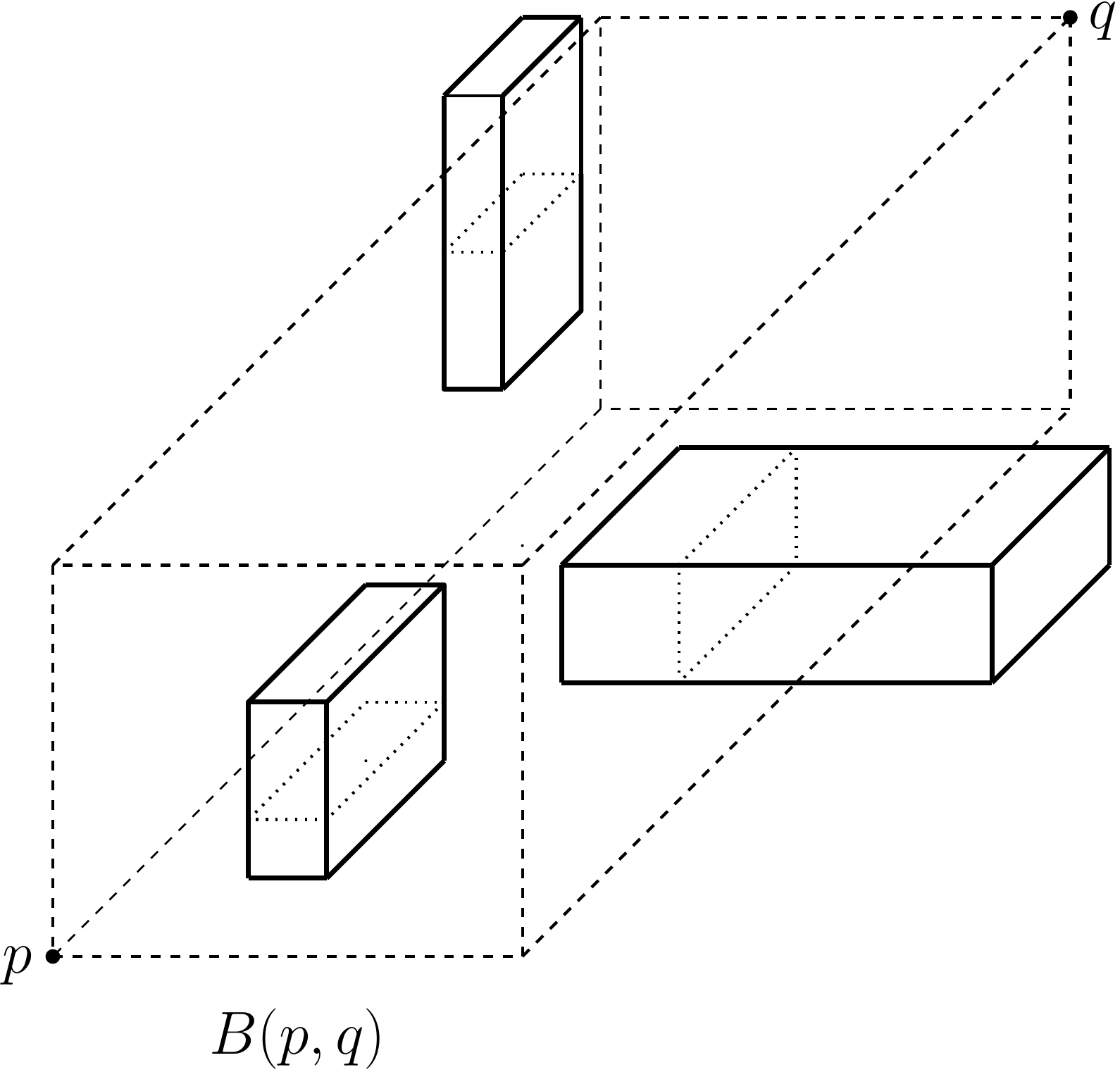}
	\end{center}
	\caption{An environment consisting of three obstacles (the solid boxes) and $\B(p,q)$ (the dashed box). The intersection of each obstacle and the boundary of $\B(p,q)$ is illustrated by dotted lines.}
	\label{bpq}
\end{figure*}

Usually $\sigma(p,q)$ is measured in the Euclidean norm (or the $L_2$ norm).  Since we will deal with the axis-parallel boxes as obstacles, it is more convenient that $\sigma(p,q)$ is measured in the Manhattan norm (or the $L_1$ norm)---several works have been devoted to computing the shortest rectilinear geodesic path in the presence of axis-parallel boxes \cite{choi95}\cite{choi96}\cite{deb91}. Let $L_i(p,q)$ ($i=1,2$) denote the $L_i$ distance of two points $p$ and $q$ in the absence of obstacles. Since for any $p,q \in \Reals^3$ we have $1/\sqrt{3}L_1(p,q) \leq L_2(p,q) \leq L_1(p,q)$, we can easily get the following observation.

\begin{observation} \label{obs:L1L2}
If $\graph$ is a $t$-spanner for the metric $\M=(P, \dist_\M)$ where distances are measured in the $L_1$ norm, then
$\graph$ is a $\sqrt{3}t$-spanner for the metric $\M=(P, \dist_\M)$ when distances are measured in the $L_2$ norm.
\end{observation}
Therefore,  from now on, we can assume that every distance is measured in the $L_1$ norm.

To warm up, we first show an obvious lower bound on the spanning ratio. Indeed, we show that there is a configuration of points and obstacles such that any non-complete graph has a spanning ratio larger than $2-\eps$.
\begin{lemma}
For any $\eps>0$, there exists a set $P$  of $n$ point in $\Reals^3$ amid a bounded number of axis-parallel boxes as obstacles such that every non-complete graph with vertex set $P$ has the spanning ratio larger than $2-\eps$.
\end{lemma}
\begin{proof}
Suppose $P$ is a set of $n$ points on the $x$-axis such that the distance of the rightmost point and the leftmost point is less than $\eps$.
Imagine there is an obstacle  between any two consecutive points where the $x$-side, $y$-side and $z$-side of each obstacle are of length $0$, $s$ and $s$ for some constant $s>0$, respectively, and the mass center of each obstacle is on the $x$-axis (see Fig. \ref{planes}). 
It is easy to see for any two points $p,q \in P$, we have $s\leq \sigma(p,q)\leq s+\eps$.
Consider a non-complete graph $\graph$ with vertex set $P$, and let $p$ and $q$ be two points in $P$ that are not connected by an edge in $\graph$. Then, $\dist_\graph(p,q)\geq 2s$. Therefore, the spanning ratio is at least $\frac{2s}{2+\eps}$ which is larger than $2-\eps$ if we set $s> 2-\eps^2/2$.
\end{proof}

\begin{figure*}
	\begin{center}
		\includegraphics[scale=0.75]{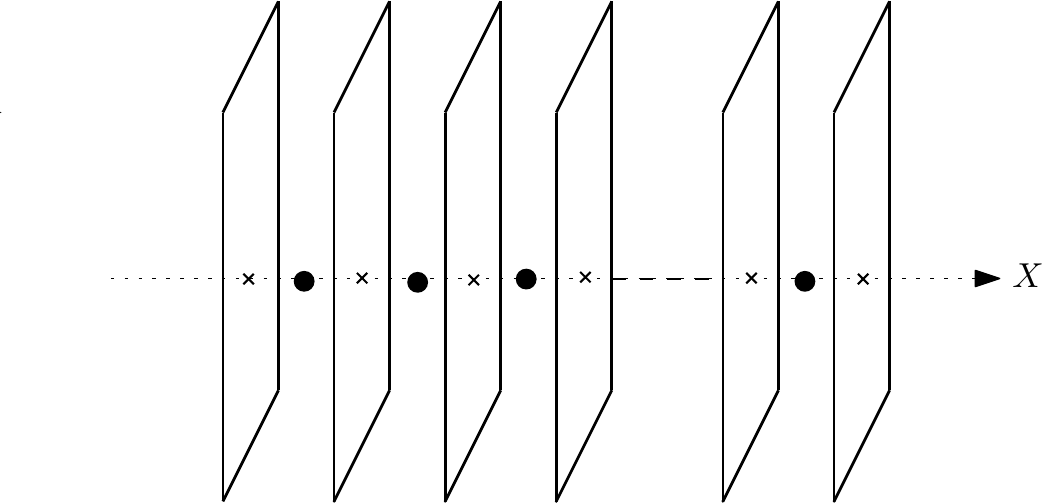}
	\end{center}
	\caption{A configuration not admitting spanning ratio less than $2-\eps$ with $o(n^2)$ edges}
	\label{planes}
\end{figure*}

Next, we explain how to construct an 8-spanner in the $L_1$ metric. We start with the main tool of our spanner construction that we believe is of independent interest.

\begin{lemma} \label{lem:main}
For any two points $p, q \in \Reals^3$ and any point $o$ inside (or on the boundary of) $\B(p,q)$, we have
$\sigma(p,o)+\sigma(o,q) \leq 4 \cdot  \sigma(p,q) $.

\end{lemma}
\begin{proof}
Let $g$ be a geodesic path from $p$ to $q$. See this and the other notion  in Fig. \ref{box1}---since the $3d$ illustration is a bit confusing, we use a $2d$ illustration to denote our notation. The length of $g$ of course is $\sigma(p,q)$.  We claim (and prove later) that there is a point $r$ on $g$ such that 
$\sigma(o,r) \leq (3/2) \cdot\sigma(p,q)$. By the triangle inequality, we know that $\sigma(p,o) \leq \sigma(p,r)+\sigma(r,o)$ and 
$\sigma(o,q) \leq \sigma(o,r)+\sigma(r,q)$. Summing up these two inequalities  and using the claim, we have 
$$\sigma(p,o)+\sigma(o,q) \leq 4  \cdot \sigma(p,q)$$

\begin{figure*}
	\begin{center}
		\includegraphics[scale=0.7]{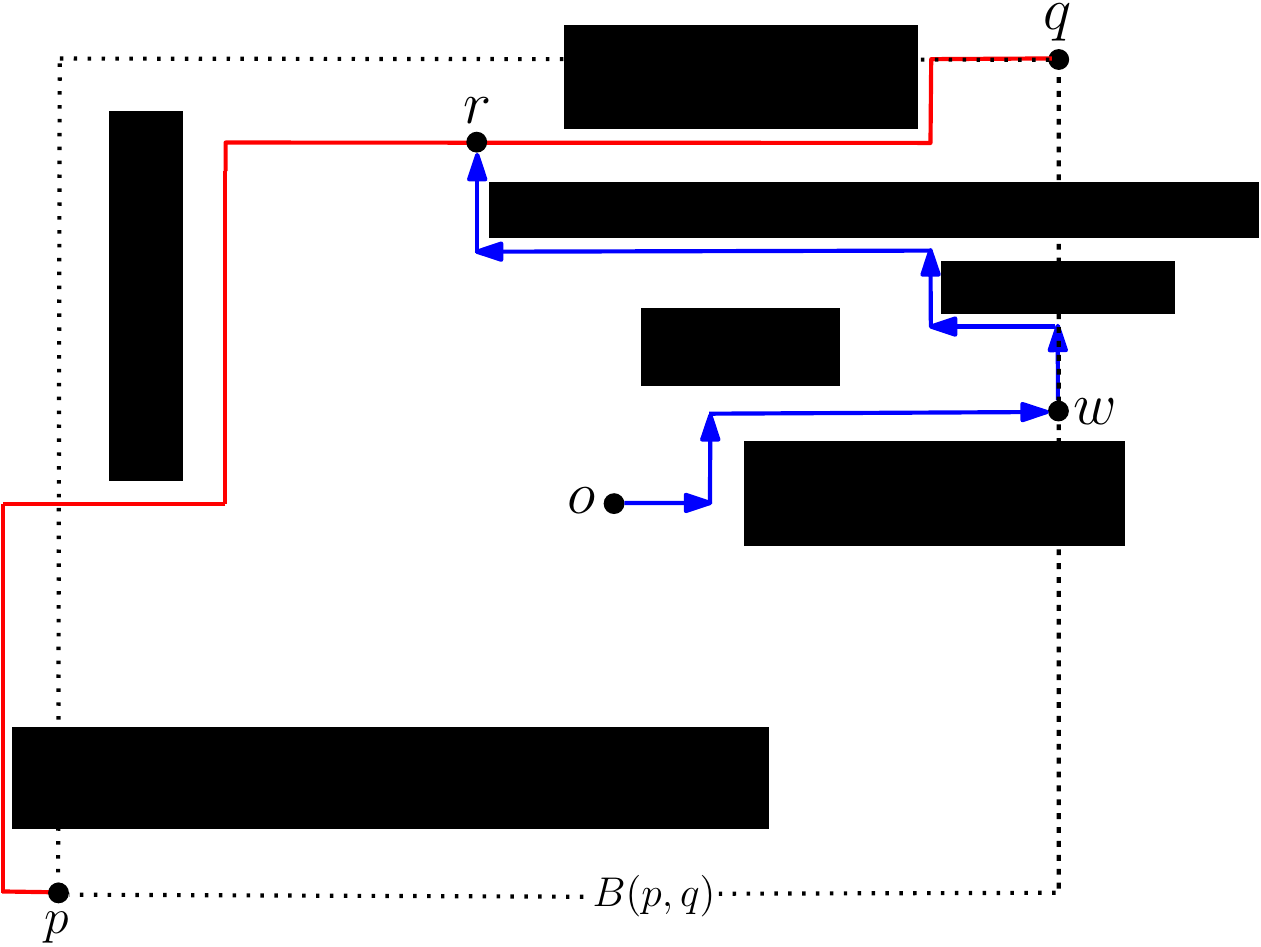}
	\end{center}
	\caption{A $2d$ illustration for the proof  Lemma \ref{lem:main}. The path $g$, a  geodesic  path from $p$  to $q$, is illustrated by a red polyline.}
	\label{box1}
\end{figure*}

Then, it remains to prove the claim. W.l.o.g., assume that $x(p) \leq x(q), y(p) \leq y(q)$ and  $z(p) \leq z(q)$ where $x(.), y(.)$ and $z(.)$ denote the $x, y$ and $z$-coordinates, respectively.  We know that $L_1(p,q) = L_1(p,o)+L_1(o,q)$---we recall that $L_1(p,q)$ denotes the $L_1$ distance of $p$ and $q$ in the absence of obstacles and obviously $L_1(p,q) \leq \sigma(p,q)$. Therefore, one of $L_1(p,o)$ and $L_1(o,q)$ is at most $(1/2)\cdot L_1(p,q)$. W.l.o.g., assume $L_1(o,q) \leq (1/2)\cdot L_1(p,q)$. 

Now, move a  point $s$  continuously on the positive direction of the $x$-axis, $y$-axis or $z$-axis, starting at $o$ and avoiding obstacles, until $s$ reaches a point $w$ on one of the edges of $\B(p,q)$ incident to $q$. This is always possible as at each position of $s$, there is at least one direction out of the three directions (the positive $x$-axis, $y$-axis and $z$-axis) such that we can move $s$ on that direction (i.e., $s$ is not blocked by any obstacle or any side of $\B(p,q)$ in that direction).  For example, assume $s$ moves on the positive direction of the $x$-axis and hits an obstacle. The other two directions are not blocked by the obstacle and $s$ can continue its motion on either  the positive $y$-axis or the positive $z$-axis (note that obstacles are axis-parallel boxes).
W.l.o.g., assume $w$ is on the edge parallel to the $z$-axis. The length of the path traveled by $s$ from $o$ to $w$ is $L_1(o,w)$ as the $x$, $y$ and $z$-coordiantes of $s$ never decrease in its motion. This implies $\sigma(o,w)\leq L_1(o,w)$, and since we always have $\sigma(o,w)\geq L_1(o,w)$, we get $\sigma(o,w)= L_1(o,w)$.  Since $w$ is on the boundary of $\B(o,q)$ (the box with $oq$ as the main diagonal), we have  $L_1(o,w) \leq L_1(o,q)$. Therefore,  $\sigma(o,w)=L_1(o,w) \leq L_1(o,q)\leq (1/2)\cdot L_1(p,q) \leq (1/2)\cdot \sigma(p,q)$.

Now, simultaneously move  continously  two points $s_1$ and $s_2$, respectively starting  at $q$ and $w$, as follows---we recall that $x(w)=x(q), y(w)=y(q)$ and $z(w)\leq z(q)$. The point $s_1$ moves on $g$,   the  geodesic path from $q$ to $p$ that we fixed it at the begining of the proof, and both $s_1$  and $s_2$ follow the  rules described below  in the given order (i.e., we never run Step (ii) unless Step (i) is impossible to run, and we never run Step (iii) unless Steps (i)  and (ii) are impossible to run):

\begin{itemize}
\item[(i)] If  $s_2$ is free to move on the positive direction of the $z$-axis, $s_2$ moves on the  positive direction of the $z$-axis and $s_1$ stays unmoved.
\item[(ii)] If $s_1$ moves on the positive or negative direction of the $z$-axis , $s_2$ stays unmoved. 
\item[(iii)] Otherwise, $s_2$ follows $s_1$ in the direction of   the $x$-axis or $y$-axis (i.e., their $x$ and $y$-coordinates are the same during their motions). 
\end{itemize}
If we are at Step (iii) during the above process,  $s_2$ can freely follow $s_1$ in the direction of the $x$-axis and $y$-axis without any obstacle blocking it.  The reason is that when we are at Step (iii), we know $s_2$ is blocked  in the positive direction of the $z$-axis. So, the other two directions are free for $s_2$ to move. 

Initially,	 $z(q)=z(s_1) \geq z(s_2)=z(w)$.  At some time $s_1$ and $s_2$  must collide at a point $r$ on $g$ due to the following reasons: (i) $s_2$ moves on the positive direction of  the $z$-axis, (ii) $s_1$ at some time reaches a point whose $z$-coordinate is $z(w)$ (note that $s_1$ moves on $g$ and finally it must reach $p$ whose $z$-coordinate is at most $z(w)$), and (iii) $s_1$ and $s_2$ keep their $x$ and $y$ coordinates the same throughout the process.

The length of the path traveled by $s_2$ from $w$ to $r$  is at most the length of $g$ (i.e., $\sigma(p,q)$). The reason is that $s_2$ behaves like $s_1$ on the $x$ and $y$-axis and the $z$ distance traveled by $s_2$  is at most the $z$ distance traveled by any moving point from $p$ to $q$ on $g$ (note that $s_2$ only moves in the positive direction of the $z$-axis  and starts its motion from $w$ whose $z$-coordinate is at least $z(p)$). 
 
Now consider the paths described above from $o$ to $w$ and then from $w$ to $r$.  The lengths of these paths are at most $(1/2)\cdot\sigma(p,q)$ and $\sigma(p,q)$, respectively. Therefore, $\sigma(o,r) \leq (3/2)\cdot \sigma(p,q)$ as we claimed.
\end{proof}
\paragraph{Remark.} In the 2-dimensional space where obstacles are rectangles, we can easily show that $\sigma(p,o)+\sigma(o,q) \leq 3 \cdot \sigma(p,q) $ and this is tight. Similar to the proof given in the lemma, we can show that there exists a point $r$ on $g$ such that $\sigma(o,r) \leq \sigma(p,q)$. The proof is much simpler\ unlike the proof given in the lemma. To this end, set $o$ to be the origin. The geodesic path $g$ definitely intersects the $x$-axis and the $y$-axis. W.l.o.g, assume the case illustrated in Fig.~\ref{box2}(a) happens. If we start continuously moving from $o$  to left or up avoiding obstacles, we definitely reach a point $r$ on $g$. It is easy to see the length of this path is at most $\sigma(p,q)$. Fig. \ref{box2}(b) shows that the inequality is tight. 

\begin{figure*}
	\begin{center}
		\includegraphics[width=1.0\textwidth]{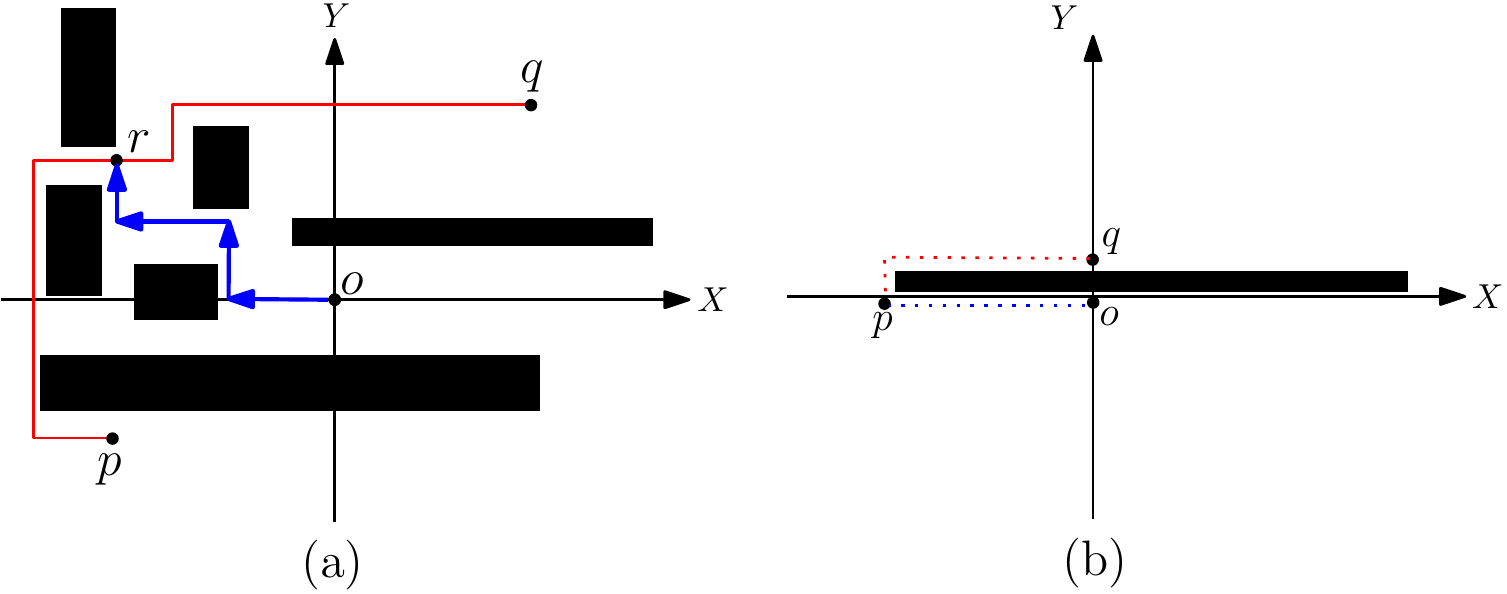}
	\end{center}
	\caption{(a) An instance of the 2-dimensional variant of the problem, (b) A tight example}
	\label{box2}
\end{figure*}

Before we explain our $t$-spanner construction, we introduce the final tool used in our construction, namely the cone-separated pair decomposition \cite{ab-ks-11}. Let $\C$ be the set of four cones defined by the $xy$, $xz$ and $yz$-planes (with an apex at  the intersection point of these three planes) that are above the $xy$-plane. For a cone $\mu \in \C$ and any point $p\in \Reals^3$, let $\mu(p)$ denote the translated copy of $\mu$ whose apex coincides with $p$. Also let $\bar{\mu}(p)$ be the reflection of $\mu(p)$ about $p$. For a cone $\mu \in \C$ and a set $P$ of $n$ points in $\Reals^3$, the cone-separated pair decomposition is defined as follows:

\begin{definition} \cite{ab-ks-11}
A \emph{cone-separated pair decomposition}, or \emph{CSPD} for short, for $P$ with respect to $\mu$
is a collection $\Psi_\mu := \{(A_1,B_1),\ldots,(A_m,B_m)\}$ of pairs of subsets from $P$
such that
\begin{enumerate}
\item[(i)] For every two points $p,q\in P$ with $q\in\mu(p)$,
           there is a unique pair $(A_i, B_i)\in \Psi_\mu$ such that $p \in A_i$ and $q \in B_i$.
\item[(ii)] For any pair $(A_i, B_i)\in \Psi_\mu$ and every two points $p \in A_i$ and $q \in B_i$,
            we have $q\in \mu(p)$ and, hence, $p\in\bar\mu(q)$.
\end{enumerate}
\end{definition}

 Abam and de Berg \cite{ab-ks-11} showed that  a CSPD $\Psi_\mu := \{(A_1,B_1),\ldots,(A_m,B_m)\}$ can be constructed with the property that $\sum_{i=1}^m |A_i|+|B_i| = O(n\log^3n)$.

 \smallpar{Spanner construction.} Next we show how to compute a spanner~$\graph=(P, E)$ for the metric space $\M = (P, \dist_\M)$
 for the set $P$ of $n$ points in $\Reals^3$ amid axis-parallel boxes as obstacles. 
 
 \begin{enumerate}
\item For each of the four cones $\mu \in \C$, we construct a CSPD $\Psi_\mu$.
\item For each pair $(A, B) \in \Psi_\mu$, let $o$ be a point such that $\mu(o)$ and $\bar\mu(o)$ contain all points of $B$ and all points of $A$, respectively. The point $o$ may be inside an obstacle $\obs$.  Let $o_{x^+}$ and $o_{x^-}$ be the extreme points on $\obs$ (recall that $\obs$ is an axis-parallel box) respectively in the positive $x$-axis and the negative $x$-axis such that $y(o)=y(o_{x^+})=y(o_{x^-})$ and $z(o)=z(o_{x^+})=z(o_{x^-})$. If $o$ is not inside any obstacle, then $o_{x^+}=o_{x^-}=o$. We define $o_{y^+}, o_{y^-}, o_{z^+}$ and $o_{z^-}$ in a similar way for the $y$-axis and $z$-axis.
For each point in the set $\{o_{x^+}, o_{x^-}, o_{y^+}, o_{y^-}, o_{z^+}, o_{z^-}\}$, for instance $o_{x^+}$, we find a point $p \in A\cup B$ whose geodesic distance to $o_{x^+}$ (i.e., $\sigma(p, o_{x^+})$) is minimum. For each 
$q \in A\cup B$, we add the edge $(p,q)$ to our spanner $\graph$. We recall that the weight of edge $(p,q)$  comes from the metric space $\M$ which is the geodesic distance from $p$ to $q$. 
\end{enumerate}

\begin{lemma}\label{lem:spanner}
	The construction above gives an $8$-spanner with respect to the geodesic distance.
	Moreover, the spanner has $O(n\log^3 n)$ edges.
\end{lemma}
\begin{proof}
The number of edges we add to the spanner  for each pair $(A,B)$ of a CSPD is at most $6(|A|+|B|)$ (6 comes from  
$|\{o_{x^+}, o_{x^-}, o_{y^+}, o_{y^-}, o_{z^+}, o_{z^-}\}|$). Since $|\C|=4$ and for each $\mu \in \C$, $\sum_{(A,B) \in \Psi_\mu} |A|+|B| = O(n\log^3 n)$, the total number of edges we have in our spanner is $O(n\log^3 n)$.
 
We next bound the spanning ratio. Let $p,q$ be two arbitrary points in $P$. There is $\mu \in \C$ and a pair $(A, B) \in \Psi_\mu$ such that $p \in A$ and $q \in B$.
Consider the point $o$ and the obstacle $\obs$ and the set  $\{o_{x^+}, o_{x^-}, o_{y^+}, o_{y^-}, o_{z^+}, o_{z^-}\}$ as defined in the construction. Note that if $o$ is not inside any obstacle, all these six points are equal to $o$.

We first prove that at least one point in the set $\{o_{x^+}, o_{x^-}, o_{y^+}, o_{y^-}, o_{z^+}, o_{z^-}\}$ is inside $\B(p,q)$. W.l.o.g., assume $x(p) \leq x(q), y(p) \leq y(q)$ and $z(p) \leq z(q)$ where $x(.)$, $y(.)$ and $z(.)$ denote the $x$ $y$ and $z$-coordiantes, respectively. It is obvious that $o$ is inside $\B(p,q)$ (notice that $q \in \mu(o)$ and $p \in \bar\mu(o)$). If $o_{x^+}, o_{x^-} \not\in \B(p,q)$, then $x(o_{x^+}) > x(q)$ and $ x(o_{x^-}) < x(p)$. Similar inequalities hold for the $y$-axis and the  $z$-axis if $o_{y^+}, o_{y^-}, o_{z^+}, o_{z^-} \not\in \B(p,q)$. All these together imply both $p$ and $q$ must be inside  $\obs$ which is  a contradiction.

W.l.o.g., assume  $o_{x^+} \in \B(p,q)$. Let $r \in A \cup B$ be the point such that $\sigma(r,o_{x^+})$ is minimum among all points in $A \cup B$. We know that edges $(p,r)$ and $(q,r)$ exist in our spanner. Therefore,
	\[
	\begin{array}{lll}
	\dist_\graph(p,q) & \leq & \sigma(p,r)+\sigma(r,q) \\[1mm]
	&   \leq & (\sigma(p, o_{x^+})+\sigma(o_{x^+},r))+(\sigma(r, o_{x^+})+\sigma(o_{x^+},q))   \\[1mm]
	&   = & (\sigma(p, o_{x^+})+ \sigma(o_{x^+},q)) + 2 \sigma(o_{x^+},r) \\[1mm]
    & \leq & 2 (\sigma(p, o_{x^+})+ \sigma(o_{x^+},q)) \\[1mm]
    & \leq & 8 \cdot \sigma(p,q)
	\end{array}
	\]
The last inequality is obtained from Lemma \ref{lem:main}.
\end{proof}
\smallpar{Remark.} We  only
focused on proving the existence of the spanner and the construction time of the spanner was not our desire. But it is easy to see that the spanner can be computed in polynomial time based on $n$ and $m$ where $m$ is the number of obstacles---see \cite{ab-ks-11} and \cite{choi95}
for how to compute CSPDs and the shortest $L_1$  geodesic paths in polynomial time, respectively.

Putting all this together, we get our main theorem.
\begin{theorem}
Suppose $P$ is a set of $n$ point in $\Reals^3$ amid a bounded number of obstacles, and obstacles are axis-parallel boxes.
$P$ admits an $8\sqrt{3}$-spanner with $O(n\log^3 n)$ edges with respect to the geodesic distance where distances are measured in the $L_2$  norm. 
\end{theorem}
\section{Conclusion}
We have shown that any set of $n$ points in $\Reals^3$ amid axis-parallel boxes as obstacles  admits a geodesic spanner
of spanning ratio $8\sqrt{3}$ with $O(n\log^3 n)$ edges. This is the first geodesic spanner for
points in $\Reals^3$ amid obstacles. We leave  designing a spanner with fewer edges and smaller spanning ratio as an open problem for future research. We study the problem when obstacles are axis-parallel boxes. It is worth studying the problem when obstacles are convex.





\end{document}